\documentclass[reprint,amsmath,amssymb,aps,nofootinbib]{revtex4-1}
\usepackage[utf8]{inputenc}

\usepackage{graphicx}
\usepackage{color}
\usepackage{hyperref}
\usepackage{comment}

\usepackage{amsmath}  
\usepackage{amsfonts}
\usepackage{amssymb}
\usepackage{amsthm}

\usepackage{qcircuit}
\usepackage{braket}

\newtheorem{fact}{Fact}
\newtheorem{definition}[fact]{Definition}
\newtheorem{theorem}[fact]{Theorem}
\newtheorem*{theorem*}{Theorem}
\newtheorem{lemma}[fact]{Lemma}

\newcommand{\eps}{\varepsilon}
\newcommand{\patbox}[1]{\vspace{2mm}\noindent\fbox{\parbox{0.47\textwidth}{\hspace{1mm}\parbox{0.45\textwidth}{\hspace{1mm}\\#1\vspace{1.5mm}}}}\\}

\begin{document}

\title{Quantum Algorithms for Estimating Physical Quantities using Block-Encodings}

\author{Patrick Rall}
\affiliation{Quantum Information Center, University of Texas at Austin}

\date{\today}

\begin{abstract}
We present quantum algorithms for the estimation of $n$-time correlation functions, the local and non-local density of states, and dynamical linear response functions. These algorithms are all based on block-encodings - a versatile technique for the manipulation of arbitrary non-unitary matrices on a quantum computer. We describe how to `sketch' these quantities via the kernel polynomial method which is a standard strategy in numerical condensed matter physics. These algorithms use amplitude estimation to obtain a quadratic speedup in the accuracy over previous results, can capture any observables and Hamiltonians presented as linear combinations of Pauli matrices, and are modular enough to leverage future advances in Hamiltonian simulation and state preparation.\end{abstract}

\maketitle

\section{Introduction}

A central goal of quantum algorithms is to aid in the study of large quantum systems. It is well established, for example, that quantum computers can simulate the dynamics of most Hamiltonians of interest \cite{1906.07115}.  Hamiltonian simulation algorithms, sometimes combined with the quantum Fourier transform, have led to quantum algorithms for some physical quantities, including correlation functions \cite{1401.2430} and dynamical linear response functions \cite{1804.01505}. Both of these examples are crucial for the understanding of phenomena in condensed matter physics like electron and neutron scattering \cite{west, sears}, conductivity and magnetization \cite{diventra}.

Recent work in Hamiltonian simulation has yielded algorithms with exponential improvements in accuracy \cite{1511.02306} over Trotterization and guarantee linear scaling with the simulation time \cite{1906.07115}. The strategies employed by these works can be neatly encompassed in terms of `block-encodings' - a tool that allows quantum computers to represent non-unitary matrices. These block-encodings can be built using linear combinations of unitaries (LCUs) \cite{1501.01715, 1511.02306} and manipulated using quantum singular value transformation \cite{1806.01838}. In addition to providing new and better algorithms, block-encodings provide an intuitive and powerful framework for performing linear algebra on a quantum computer.

In this work we use block-encodings along with amplitude amplification \cite{0005055, 1908.10846, 1912.05559} to construct quantum algorithms for some physical quantities: $n$-time correlation functions, the local and non-local density of states, and dynamical linear response functions. These algorithms are more versatile than previous works \cite{1401.2430, 1804.01505} in that they can compute more general versions of the functions with greater accuracy. 

The local and non-local density of states and linear response functions are all functions of the energy $f(E)$. We are usually interested in obtaining the general shape of $f(E)$ over a range of energies, i.e. in obtaining a `sketch' of $f(E)$. We show how to perform two sketching strategies from modern classical numerical condensed matter physics \cite{0504627, 1101.5895, 1811.07387}. First, we show how to compute integrals of $f(E)$ over a range of energies: $\int_{E_A}^{E_B} f(E) dE$. Second, we show how to compute the moments of a Chebyshev expansion of $f(E)$: briefly assuming $|E| \leq 1$ for ease of explanation, if $T_n(E)$ is the $n$'th Chebyshev polynomial of the first kind, then we show how to compute constants $c_n$ such that
\begin{align}
    f(E) \approx \frac{1}{\pi\sqrt{1-E^2}} \cdot \sum_{n=0}^{N} c_n T_n(E).
\end{align}
This procedure is known as the kernel polynomial method \cite{0504627} and is intuitively similar to sketching a function by computing the first few coefficients in its Fourier series.  Very recent work \cite{2004.04889} shows how similar methods can also perform point-estimates of the density of states by approximating a delta function with a polynomial close to a narrow Gaussian.

Algorithms that compute physical quantities often face barriers from complexity theory, since computing expectations of observables on ground states of Hamiltonians is $\mathsf{QMA}$-complete \cite{0406180}. This remains true even when severe restrictions are placed on Hamiltonians \cite{1212.6312}. For this reason we employ strategies that sidestep these barriers. For correlation functions, we do not provide algorithms for preparing ground states or other states of interest, since the best algorithms for their preparation must use properties of the particular Hamiltonian in question.  Evaluating the density of states at particular energies is \#$\mathsf{P}$-complete \cite{1010.3060}, but sketching the density of states via integrals and Chebyshev expansions is in $\mathsf{BQP}$.

The structure of our paper is as follows. In section~II. we review block-encoding techniques. In section~III. we employ these techniques to study $n$-time correlation functions. If we have a set of observables ${O_i}$ and times $\{t_i\}$ we compute expectations of the form
\begin{align}\left\langle O_1(t_1) O_2(t_2) ... \right\rangle\end{align}
employing the Heisenberg picture. In section~IV. we outline quantum singular value transformation and tools for computing Chebyshev moments and integrals over energy intervals. In section~V. we employ these techniques to compute the density of states and the local density of states. If $H$ has eigenvalues $\{E_i\}$ and dimension $D$ then the density of states is:
\begin{align} \rho(E) = \frac{1}{D} \sum_i \delta(E_i - E).\label{eqn:dos}\end{align}

Furthermore, say $H$ is a Hamiltonian describing a particle with some set of positions $\{ \vec r \}$ and position eigenstates $\{\ket{\vec r}\}$. If the eigenvectors of $H$ are $\{\ket{\psi_i}\}$, then the local density of states is:
\begin{align} \rho_{\vec r}(E) = \sum_i \delta(E_i - E) |\braket{\psi_i | \vec r}|^2  \label{eqn:ldos}\end{align}

Finally in section~VI. we show how to sketch linear response functions of the form
\begin{align} A(E) = \left\langle  B \delta(E - H + E_0)C  \right\rangle \label{eqn:lres} \end{align}
where $E_0$ is the ground state energy of $H$ and $B,C$ are some observables. In the appendix we show how to construct optimal polynomial approximations to the window function, which we require to compute integrals of $\rho(E), \rho_{\vec r}(E)$ and $A(E)$.

\section{Block-Encoding Techniques}

Block encodings allow quantum computers to perform manipulations with non-unitary matrices. If $A$ is \emph{any} matrix with $|A| \leq 1$ where $|A|$ is the largest singular value, then a block-encoding is a unitary $U_A$ such that $A$ occupies the top left corner of $U_A$:
\begin{align}U_A = \begin{bmatrix} A  & \hspace{1mm}\cdot\hspace{1mm} \\ \cdot & \cdot \end{bmatrix}\end{align}
    Below we give a more formal definition involving an explicit Hilbert space $\mathcal{H}$ for $A$ and an ancillary Hilbert space $\mathbb{C}^k$ for postselection\footnote{In the general case when $A$ is a rectangular matrix that maps $\mathcal{H}\to\mathcal{H}'$ then the input ancilla space $\mathbb{C}^k$ and output ancilla space $\mathbb{C}^l$ must be chosen so that $\mathcal{H}\otimes \mathbb{C}^k$  and $\mathcal{H}'\otimes \mathbb{C}^l$ have the same dimension. For this paper we assume that $A$ is square so we can pick $l=k$.}. We also give a notion of accuracy and a notion of scaling to allow for $|A| > 1$. The number of qubits needed to realize these spaces is bounded by the circuit complexity of $U_A$. We denote the computational basis for ancillary Hilbert spaces $\mathbb{C}^k$ by $\{\ket{0}_k, \ket{1}_k, \ldots\}$.

\begin{definition} \label{def:block} Say $A$ is a matrix on $\mathcal{H}$ with $|A| \leq \alpha$. A unitary $U_A$ on $\mathbb{C}^k \otimes \mathcal{H} $ is an $\eps$-accurate $\alpha$-scaled $Q$-block-encoding of $A$ if $U_A$ is implementable using $Q$ elementary gates and for some $k$ we have
\begin{align}|A/\alpha - (\bra{0}_k\otimes I)U_A(\ket{0}_k\otimes I)| \leq \eps. \end{align}
If `$\eps$-accurate' is omitted then 0-accurate (exact) is implied, and if `$\alpha$-scaled' is omitted then 1-scaled is implied.
\end{definition}

In our work we will only be interested in block-encodings of products of observables, so $A$ will be square and often Hermitian. The Pauli matrices are a basis for Hermitian matrices, but since they are also unitary they have trivial ($U_P = P$) block-encodings. A key property of block-encodings is that a quantum computer can easily prepare products and linear combinations of them.

\begin{lemma}  \label{lemma:lcu} Say the matrices $\{A_i\}$ each have $\alpha_i$-scaled  $Q_i$-block-encodings. Then:
\begin{enumerate}
 \item the product $\prod_i A_i$ has a $\left(\prod_i \alpha_i\right)$-scaled  $O\left(\sum_i Q_i \right)$-block-encoding, and
 \item for any $\beta_i \in \mathbb{C}$ the linear combination $\sum_i \beta_i A_i$ has a $\left(\sum_i \alpha_i |\beta_i|\right)$-scaled $O\left(\sum_i Q_i\right)$-block-encoding.
\end{enumerate}
\end{lemma}
\begin{proof}  A complete construction and analysis of these circuits is given in \cite{1806.01838}, although the core techniques were put forth earlier \cite{1501.01715, 1511.02306}. The construction of block-encodings of products is rather trivial, and we give a brief sketch of the proof that a linear combination of Pauli matrices $O = \sum_{i=1}^k\beta_i P_i$ has a $O\left(\sum_i\beta_i\right)$-scaled $O(k)$-block-encoding $U_O$:
\begin{align}
V_\beta \ket{0}_k &:= \frac{1}{\sqrt{\sum_i |\beta_i|}}\sum_{i=1}^k \sqrt{|\beta_i|} \ket{i}_k \\
V_{P} &:= \sum_{i=1}^k \ket{i}_k\bra{i}_k \otimes \frac{\beta_i}{|\beta_i|}  P_i \\
U_{O} &:= (V_\beta^\dagger \otimes I) V_P (V_\beta \otimes I)
\end{align}
The gate complexity is dominated by $V_P$ with complexity $O(k)$. Generalizing to non-trivial block-encodings involves swapping $P_i$ with $U_{A_i}$ and dealing with the control registers.
\end{proof}

Lemma~\ref{lemma:lcu} has the crucial consequence that the vast majority of Hamiltonians in physics have efficient block-encodings, since they can be written as linear combinations of not too many Pauli matrices. In these cases we have $k,\alpha \in O(\text{poly}(n))$ where $n$ is the number of qubits required to encode $\mathcal{H}$.

The algorithms in this work construct block encodings of a desired $A$ and estimate $\text{Tr}(A \rho)$ for some given $\rho$. To do so we assume that there is a unitary that prepares a purification of $\rho$, which is any pure state such that $\rho$ can be obtained by tracing out some ancillary space $\mathbb{C}^l$.

\begin{definition} \label{def:prep} Let $\rho$ be a density operator on $\mathcal{H}$ and let $\ket{\textbf{0}}$ be some easy-to-prepare state in $\mathcal{H}$. A unitary $U_\rho$ on $\mathcal{H}\otimes \mathbb{C}^l$ for some $l$ is an $R$-preparation-unitary of $\rho$ if we have
\begin{align} \rho = \text{Tr}_{\mathbb{C}^l}  \left( \ket{\rho}\bra{\rho} \right),  \end{align}
where $\ket{\rho} = U_\rho(\ket{\mathbf{0}}\ket{0}_l)$ and $U_\rho$ is implementable using $R$ elementary gates.
\end{definition}

Often we are interested in correlation functions and linear response with respect to ground states or thermal states of some Hamiltonian. Depending on the situation performing state preparation can be an extremely difficult computational task, and the identification of specific practical situations where state preparation is easy is an area of active research \cite{1609.07877}. We consider the problem of state preparation itself out of scope for this work, but aim to present our algorithms in an abstract manner to maximize their versatility and permit the leveraging of future results. We do point out the existence of the following generic tool for constructing thermal states.

\begin{lemma} \label{lemma:thermal} Let $H$ be a Hamiltonian on a $D$-dimensional Hilbert space with an $\alpha$-scaled $Q$-block-encoding. Then for any $\beta \geq 0$ there exists an $R$-preparation unitary for a state $\eps$-close in trace distance to the thermal state $e^{-\beta H} / Z$ where $Z = \text{Tr}(e^{-\beta H})$ and:
\begin{align} R \in O\left( Q\alpha \cdot \sqrt{\frac{D\beta}{Z}}  \log\left(\sqrt{\frac{D}{Z}}\frac{1}{\eps}\right)\right) \end{align}
\end{lemma}
\begin{proof} This is the main result of \cite{1603.02940}, combined with the newer Hamiltonian simulation results of \cite{1610.06546, 1606.02685} with corrections from \cite{1806.01838}. Briefly, the strategy is to construct a block-encoding of $e^{-\beta H /2}$ from $e^{iHt}$ using the Hubbard-Stratonovich transformation, and multiply it onto a purification of the maximally mixed state using a strategy called robust oblivious amplitude amplification.
\end{proof}

We now show how to use amplitude estimation to estimate the expectation of block encoded observables.
\begin{lemma} \label{lemma:observ} If $A$ is Hermitian and has an $\alpha$-scaled $Q$-block-encoding and $\rho$ has an $R$-preparation-unitary, then for every $\eps,\delta > 0$ there exists an algorithm that produces an estimate $\xi$ of $\text{Tr}(\rho A)$ such that
\begin{align} |\xi - \text{Tr}(\rho A) | \leq \eps \end{align}
    with probability at least $(1-\delta)$. The algorithm has circuit complexity $O\left((R+Q) \cdot \frac{\alpha}{\eps}\log \frac{1}{\delta}\right)$.
\end{lemma}
\begin{proof} The algorithm is as follows:

\patbox{
\textbf{Algorithm: Observable Estimation} \\[2mm]
Let $\bar A = (I +  A/\alpha)/2$, and let $U_{\bar A}$ be its 1-scaled $O(Q)$-block-encoding which exists by Lemma~\ref{lemma:lcu}. Let $U_{\bar A}$ have control register dimension $k$ as in Definition~\ref{def:block}, and let $l$ and $\ket{\mathbf{0}}$ be as in Definition~\ref{def:prep}. Let:
\begin{align}
\ket{\rho} &:= U_\rho \ket{\mathbf{0}} \ket{0}_l\\
\ket{\Psi} &:= (U_{\bar A} \otimes I)  \ket{0}_k \ket{\rho}\\
    \Pi &:= \ket{0}_k\bra{0}_k \otimes \ket{\rho}\bra{\rho} 
\end{align}

Perform amplitude estimation to obtain an estimate $\xi_0$ of $|\Pi\ket{\Psi}|$ to precision $\eps/(2\alpha)$ with probability at least $(1-\delta)$. Return $\xi := (2\xi_0 + 1)\alpha$.
}

For details on how to perform amplitude estimation we refer to recent results \cite{1908.10846, 1912.05559} that avoid using the quantum Fourier transform, which was required by the traditional method \cite{0005055} from 2002. These results establish that $|\Pi\ket{\Psi}|$ can be estimated to additive error $\eps$ and probability at least $(1-\delta)$ using $O\left(\frac{1}{\eps} \log \frac{1}{\delta}\right)$ applications of a Grover operator:
\begin{align}G := -( I - 2\Pi )(I - 2\ket{\Psi}\bra{\Psi})\end{align}
This operator requires four uses of $U_\rho$ and two uses of $U_{\bar A}$, so it has circuit complexity $O(R+Q)$. This completes the runtime analysis.

Amplitude estimation estimates:
\begin{align}
|\Pi\ket{\Psi}| &= | \bra{0}_k\bra{\rho} (U_{\bar A} \otimes I)\ket{0}_k\ket{\rho} |\\ 
&=  |\bra{\rho} (\bar A \otimes I) \ket{\rho} | \\
&= | \text{Tr}( \ket{\rho}\bra{\rho} (\bar A \otimes I))| \\
&= | \text{Tr}\left(\text{Tr}_{\mathbb{C}^l}(\ket{\rho}\bra{\rho})\bar A\right) | = |\text{Tr}(\rho \bar A)|
\end{align}
Since $\bar A$ has $|\bar A| \leq 1$ its eigenvalues lie in the range $[-1,1]$, so $\bar A$ is positive semi-definite. Therefore $\xi_0$ approximates $|\text{Tr}(\rho \bar A)| = \text{Tr}(\rho \bar A) = (1 + \text{Tr}(\rho  A)/\alpha)/2$ to error $\eps/(2\alpha)$, so $\xi$ approximates $\text{Tr}(\rho A)$ to error $\eps$ as desired.
\end{proof}

In addition to providing a simple framework for manipulating observables on a quantum computer, block-encodings are often the starting point for modern Hamiltonian simulation algorithms \cite{1501.01715, 1906.07115}. Once a block-encoding of a Hamiltonian $H$ is constructed, we can apply functions to its eigenvalues using quantum singular value transformation discussed in section IV.

\section{Correlation Functions}

In this section we show how to estimate $n$-time correlation functions, improving on an algorithm presented in \cite{1401.2430}. This algorithm does not require any new technical tools. We include it primarily to illustrate how simple it is to construct algorithms for complex quantities via block-encodings. We also show how to estimate non-Hermitian block-encoded observables, a tool we will require later in section VI. Consider a system evolving under a time-independent Hamiltonian $H$. If $O_i$ is some Hermitian operator then in the Heisenberg picture:
\begin{align} O_i(t_i) := e^{iHt_i} O_i e^{-iHt_i}\end{align}

To prepare block-encodings of observables in the Heisenberg picture we leverage a modern result in Hamiltonian simulation for time-independent Hamiltonians. For simplicity we focus on time-independent Hamiltonians but there also exist block-encodings for time evolution under time-dependent Hamiltonians \cite{1906.07115, 1805.00582, 1805.00675}.

\begin{lemma} \label{lemma:hamsim} Let $H$ be a Hamiltonian on a $D$-dimensional Hilbert space with an $\alpha$-scaled $Q$-block-encoding. Then for any $t,\eps > 0$ there exists an $\eps$-accurate $T(t,\eps)$-block-encoding of $e^{i H t}$ where:
\begin{align} T(t,\eps) \in O\left( Q\alpha|t| + \frac{Q\log(1/\eps)}{\log(e + \log(1/\eps) / (\alpha |t| ))}  \right)\label{eqn:Tcomplex}\end{align}
\end{lemma}
\begin{proof} This result originated in \cite{1606.02685, 1610.06546}, but it is cleanly re-stated with minor corrections as Corollary~60 of \cite{1806.01838}.
\end{proof}

Using this result we can state and analyze the estimation algorithm.

\begin{theorem} Let:
\begin{itemize}
\item $H$ be a Hamiltonian with an $\alpha$-scaled $Q$-block-encoding,
\item $O_1,...,O_n$ be some observables with $\beta_i$-scaled $R_i$-block-encodings, 
\item $t_1,...,t_n$ be some times,
\item and $\rho$ be a state with an $S$-preparation unitary.
\end{itemize}
    Then for every $\eps,\delta > 0$ there exists an algorithm that produces estimate an estimate $\xi\in \mathbb{C}$ of $\text{Tr}\left( \rho \prod_i O_i(t_i) \right)$ to additive precision $\eps$ in the real and imaginary parts with probability at least $(1-\delta)$. It has circuit complexity $O\left((S+W) \cdot \frac{\gamma}{\eps} \log\frac{1}{\delta}\right)$ where $\gamma= \prod_i \beta_i$ and
\begin{align} W \in& O\left( \sum_{j=1}^n R_j + \sum_{j=0}^n T\left(\tau_j, \frac{\eps}{2(n+1)^2} \right)  \right)\label{eqn:Wcomplex}\\
\subset &O\left(  \sum_{j=1}^n R_j + Q \alpha \sum_{j=0}^n |\tau_j| + Q n^2 \log\left(\frac{n}{\eps}\right)\right) \label{eqn:roundedWcomplex}\end{align}
where $T(t,\eps)$ is defined in Lemma~\ref{lemma:hamsim} and $\tau_j = t_{j+1} - t_{j}$, padding the list of times with $t_0 = t_{n+1} = 0$.
\end{theorem}
\begin{proof}

The algorithm is as follows:

    \vspace{7cm} ~

\patbox{
\textbf{Algorithm: $n$-time correlation functions} \\[2mm]

Making use of $e^{-iHt_j}e^{iHt_{j+1}} = e^{iH(t_{j+1}- t_{j})} = e^{iH\tau_j} $, we rewrite the product of observables as follows:
\begin{align}
\prod_{j=1}^n O_j(t_j) &=  e^{iHt_1}O_1e^{iH(t_2-t_1)} ... O_n e^{-iHt_n}\\
&=  e^{iH\tau_0} \prod_{j=1}^n  O_j e^{iH\tau_{j}}
\end{align}
Invoking Lemma~\ref{lemma:hamsim} we obtain $\frac{\eps}{2(n+1)^2}$-accurate block-encodings of $e^{iH\tau_j}$, and we multiply them together with the block-encodings of $O_i$ using Lemma~\ref{lemma:lcu}. We obtain a $W$-block-encoding $U_\Gamma$ of an operator $\Gamma$ that approximates $\prod_i O_i(t_i)$.

Observe that $U_\Gamma^\dagger$ is a block-encoding of $\Gamma^\dagger$. This allows us to use Lemma~\ref{lemma:lcu} to construct $\gamma$-scaled $W$-block-encodings of the Hermitian and anti-Hermitian parts of $\Gamma$, as below. Then we invoke Lemma~\ref{lemma:observ} with target accuracy $\eps/2$ for each of the below to obtain $\eps$-accurate estimates of the real and imaginary parts of $\text{Tr}\left( \rho \prod_i O_i(t_i) \right)$.
\begin{align}
\Re\left(\xi\right) :=& \text{ estimate of } \text{Tr}\left( \rho \cdot \frac{\Gamma + \Gamma^\dagger}{2} \right) \label{eqn:herm}\\
\Im\left(\xi\right) :=& \text{ estimate of } \text{Tr}\left( \rho \cdot \frac{\Gamma - \Gamma^\dagger}{2i} \right) \label{eqn:antiherm}
\end{align}
}

Since the block-encodings of $e^{iH\delta t_j}$ are 1-scaled, the only contribution to $\gamma$ are the scalings of the $O_i$, so $\gamma= \prod_i \beta_i$. The runtime is dominated by the complexity $W$ of the block-encoding for $\Gamma$, which by Lemma~\ref{lemma:lcu} is clearly given by (\ref{eqn:Wcomplex}). To obtain (\ref{eqn:roundedWcomplex}) we loosely bound $1/\log(e + \log(1/\eps)/(\alpha|t|)) \leq 1$ in (\ref{eqn:Tcomplex}). This looseness overestimates the runtime in situations where $n$ is very large but the $\tau_j$ are very small.

It remains to show that $\Gamma$ is $\eps/2$-close in spectral norm to $\prod_i O_i(t_i)$, given that the block-encodings of $e^{iH\tau_j}$ are $\frac{\eps}{2(n+1)^2}$-accurate. From there the $\eps/2$-closeness of the Hermitian and anti-Hermitian parts, and the $\eps$-accuracy of the final estimates follow. In general, Lemma~54 of \cite{1806.01838} gives an argument that if $|A - U| \leq \eps_0$ and $|B - V| \leq \eps_1$ then 
\begin{align}|AB - UV| \leq \eps_0 + \eps_1 + 2\sqrt{\eps_0\eps_1} .\end{align}
Iterating this bound for a product of $\prod_{i=0}^n U_i$ where $|U_i - A_i| \leq \eps_0$ we obtain by solving a recurrence relation:
\begin{align}\left|\prod_{i=0}^n U_i - \prod_{i=0}^n A_i\right| \leq (n+1)^2 \eps_0   .\end{align}
Plugging in $\eps_0 := \frac{\eps}{2(n+1)^2} $ gives the desired upper bound of $\eps/2$.
\end{proof}

This algorithm improves over \cite{1401.2430} in several ways. First, \cite{1401.2430} restricts to Pauli observables since they are unitary. Here $O_i$ do not have to be unitary. Secondly, since we are using amplitude estimation to obtain $\xi$ we obtain a quadratic speedup in the accuracy dependence. Finally, \cite{1401.2430} restricts to Hamiltonians where exact Hamiltonian simulation can be achieved using circuit identities. Of course, for situations where these restrictions apply and the accuracy speedup can be sacrificed, their construction yields significantly smaller circuits which may be more amenable to near-term quantum computers.

\section{Integrals and Chebyshev Moments of Functions of the Energy}

In this section we introduce some tools we will require for our quantum algorithms for computing the density of states and linear response functions.

Say a Hermitian matrix $A$ has an eigenvalue-eigenvector decomposition $A = \sum_i \lambda_i \ket{\phi_i}\bra{\phi_i}$. Given a block-encoding of $A$, quantum singular value transformation allows us to construct block-encodings of $p(A) = \sum_i p(\lambda_i) \ket{\phi_i}\bra{\phi_i}$, for polynomials $p(x)$. This requires $p(x)$ to be appropriately bounded, and the complexity of the encoding scales linearly in the degree of the polynomial. This method can also be generalized to non-Hermitian $A$ with some caveats. Singular value transformation is an extremely powerful result, and is a culmination of a long line of research in quantum algorithms, presented in its full generality in \cite{1806.01838}.

\begin{lemma} \label{lemma:svt}  Let $A$ have a $Q$-block-encoding, and let $p(x)$ be a degree-$d$ polynomial satisfying $|p(x)| \leq 1$ for $x \in [-1, 1]$. Then for every $\delta > 0$ there exists a $\frac{1}{2}$-scaled $\delta$-accurate $O(Qd)$-block-encoding of $p(A)$. A description of the circuit can be computed in time $\text{poly}\left(d,\log \frac{1}{\delta}\right)$.
\end{lemma}

\begin{proof} This strategy originated in \cite{1606.02685, 1610.06546} and is developed in \cite{1806.01838} where it is formalized as Theorem~56. Calculating the circuit demands careful consideration of numerical precision. Recent work \cite{2003.02831} describes an elegant strategy for dealing with this issue. 
\end{proof}

The expressions for density of states (\ref{eqn:dos},\ref{eqn:ldos}) and linear response (\ref{eqn:lres}) are both functions of the energy $f(E)$ roughly of the form:
\begin{align} f(E) := \sum_{i}  \delta(E - E_i) \bra{\psi_i}A\ket{\psi_i}\label{eqn:fdef} \end{align}
where $\{E_i\}$ and $\{\ket{\psi_i}\}$ are the eigenvalues and eigenvectors of the Hamiltonian and $A$ is some Hermitian matrix. Rather than computing point-estimates of $f(E)$ we will be interested in computing integrals of $f(E)$ over a range $[a,b]$ as well as the moments of a Chebyshev expansion of $f(E)$. To obtain the scaling requirements of Lemma~\ref{lemma:svt} we observe that an $\alpha$-scaled block-encoding of a Hamiltonian $H$  guarantees that $|H/\alpha| \leq 1$. Rescaling $\bar a = a/\alpha$ and $\bar b = b/\alpha$, we construct a polynomial $w(x)$ that allows us to approximate integrals over the range $[\bar a, \bar b]$:

\begin{theorem} \label{thm:windowfunc}For every $\eta > 0$ and any $\bar a, \bar b$ with $-1 < \bar a < \bar b < 1$ there there exists a polynomial $w(x)$ such that for all $f(\alpha x)$ bounded by $f_\text{max}$ (defined below in (\ref{eq:boundedby})):
\begin{align}\left|\int_{-1}^{1} f(\alpha x) w(x) dx - \int_{\bar a}^{\bar b} f(\alpha x) dx \right| \leq \eta \end{align}
The polynomial has degree $d \in O( \frac{f_\text{max}}{\eta}\ln\frac{f_\text{max}}{\eta }  ) $ and satisfies the requirement $|w(x)| \leq 1$ of Lemma~\ref{lemma:svt}.
\end{theorem}
\begin{proof} There exist several strategies for constructing approximating polynomials for window and step functions, which we could adapt for our purposes via shifting and scaling \cite{dolph, 1409.3305, 1707.05391, 0604324, 1907.11748}. We adapt an elegant approach that relies on standard strategies in approximation theory discussed in \cite{0902.3757} leveraging amplifying polynomials and Jackson's theorem \cite{rivlin} which constructs a polynomial that accomplishes our requirements directly. We postpone the argument to Appendix~\ref{section:window}.
\end{proof}

Our accuracy analysis requires a bound on $f(\alpha x)$, which is a bit subtle to define since $f(\alpha x)$ is a sum of many delta functions. However, we only ever perform integrals of $f(\alpha x)$. Therefore when we say `$f(\alpha x)$ is bounded by $f_\text{max}$' we mean that for all $\bar c<\bar d$:
\begin{align}
    \int_{\bar c}^{\bar d} f(\alpha x) dx \leq f_\text{max} \cdot (\bar d-\bar c) \label{eq:boundedby}
\end{align}

The polynomial $w(x)$ immediately yields a strategy for computing integrals since the value can be expressed as a trace inner product.
\begin{align}
&\int_{a}^b f(E) dE  = \int_{\bar a}^{\bar b} f(\alpha x)\cdot \alpha dx\label{eqn:windowderiv1}\\
&\approx \alpha  \int_{-1}^{1} f(\alpha x) w(x) dx \\
&= \alpha \int_{-1}^{1}\sum_{i}  \delta(\alpha x - E_i) \bra{\psi_i}A\ket{\psi_i} w(x) dx \\
&= \text{Tr}\left( A \sum_{i} \int_{-1}^{1} \delta(x - E_i/\alpha) w(x) dx\ket{\psi_i}\bra{\psi_i} \right)\label{eqn:deltastep}\\
&=\text{Tr}\left( A \sum_{i} w(E_i/\alpha) \ket{\psi_i}\bra{\psi_i} \right)\\
&= \text{Tr}\left(A w(H/\alpha)\right)\label{eqn:windowderiv5}
\end{align}
In step (\ref{eqn:deltastep}) we used the identity $\delta(\alpha x) = \delta(x)/\alpha$. This final expression can then be estimated using Lemma~\ref{lemma:observ}.

Next we briefly outline our strategy for sketching $f(E)$ using the kernel polynomial method \cite{0504627}. A sketch $f^\text{KPM}(E)$ is a linear combination of Chebyshev polynomials of the first kind $T_n(x)$ weighted by coefficients $\mu^f_n g_n$. The $\mu^f_n$ are the Chebychev moments of $f(E)$ and the $g_n$ are $f(E)$-independent smoothing coefficients (see for example the proof of Jackson's theorem in \cite{rivlin}). Since Chebyshev expansions are performed on the domain $[-1,1]$ we calculate moments of $f(\alpha x)$ for $x \in[-1,1]$.

\begin{align}
\mu^f_n &:= \int_{-1}^1 T_n(x) f(\alpha x) dx\\
f^\text{KPM}(\alpha x) &:= \frac{1}{\pi\sqrt{1 - x^2}} \left(g_0 \mu^f_0 + 2 \sum_{n=0}^N \mu^f_ng_n T_n(x)\right)
\end{align}
For this work we concern ourselves only with estimation of $\mu^f_n$ and defer to \cite{0504627, 1811.07387} for details on how to construct $f^\text{KPM}(E)$. A similar derivation to (\ref{eqn:windowderiv1}-\ref{eqn:windowderiv5}) yields the identity:
\begin{align}
\mu^f_n &:= \int_{-1}^1 T_n(x) f(\alpha x) dx = \text{Tr}\left( A T_n(H/\alpha) \label{eqn:mudef}\right)
\end{align}
Conveniently, quantum singular value transformation is particularly simple for Chebyshev polynomials.

\begin{lemma}\label{lemma:cheby} Let $A$ have a $Q$-block-encoding. Then for every $n$ there exists an $O(nQ)$-block-encoding of $T_n(A).$
\end{lemma}
\begin{proof} This is Lemma~9 of \cite{1806.01838}.\end{proof}

Now we have all the technical tools to state the main algorithms. 

\section{Density of States}

In this section we show how to sketch the density of states (DOS):
\begin{align} \rho(E) = \frac{1}{D} \sum_i \delta(E_i - E).\end{align}
This is easily rewritten in the form in (\ref{eqn:fdef}) by choosing $A = I/D$. Following (\ref{eqn:windowderiv1}-\ref{eqn:windowderiv5}) and (\ref{eqn:mudef}) we obtain:
\begin{align}
    \int_{a}^b \rho(E) dE &\approx \text{Tr}\left( \frac{I}{D} w(H/\alpha) \right)\\
    \mu^{\rho}_n &= \text{Tr}\left( \frac{I}{D} T_n(H/\alpha) \right)
\end{align}
This argument makes use of of Theorem~\ref{thm:windowfunc} which requires a bound on $\rho(E)$. Observe that in the sense of (\ref{eq:boundedby}), $\rho(\alpha x)$ is bounded by any upper bound on the dimension of the largest eigenspace of $H$ which we call $\rho_\text{max}$.

These quantities can be estimated by leveraging the fact that $I/D$ has an $O(\log(D))$-preparation unitary.

\begin{theorem} Let $H$ have an $\alpha$-scaled $Q$-block-encoding and take any $\eps, \delta > 0$. Then:
\begin{enumerate}
 \item For any $a,b$ such that $-\alpha < a < b < \alpha$ there exists a quantum algorithm that produces an estimate $\xi$ of $\int_{a}^b \rho(E) dE$ with circuit complexity
     \begin{align}O\left(\left( Q\cdot \frac{\rho_\text{max}}{\eps}\log  \frac{\rho_\text{max}}{\eps} + \log D\right) \cdot \frac{1}{\eps} \log\frac{1}{\delta}    \right) \label{eqn:densintcomplexity}\end{align}
         and $O(\text{poly}(\rho_\text{max}/\eps))$ classical pre-processing, where $\rho_\text{max}$ is some upper bound on the dimension of the largest eigenspace of $H$. 
 \item For any $n$ there exists a quantum algorithm that produces an estimate $\zeta$ of $\mu^\rho_n$ with circuit complexity
     \begin{align}O\left(\left(Q\cdot n+ \log D \right) \cdot  \frac{1}{\eps} \log\frac{1}{\delta} \right).\label{eqn:denschebycomplexity}\end{align}
\end{enumerate}
The estimates $\xi$ and $\zeta$ have error $\eps$ with probability at least $(1-\delta)$.\label{thm:dosalg}
\end{theorem}
\begin{proof} Observe that a preparation unitary for $I/D$ simply prepares a Bell state on $\mathcal{H}\otimes \mathcal{H}$, call it $\ket{\text{Bell}(\mathcal{H})}$. If $\mathcal{H}$ is encoded as some subspace of a $n$-qubit system where $n = \lceil \log_2(D)\rceil$ then $\ket{\text{Bell}(\mathcal{H})}$ can be obtained from $\ket{\text{Bell}(\mathbb{C}^{2^n})}$ via amplitude amplification. This procedure can be made exact via the following standard trick involving an ancilla qubit. Observe that \begin{align}\beta := \braket{\text{Bell}(\mathbb{C}^{2^n})|\text{Bell}(\mathcal{H})} = \sqrt{D / 2^n}\end{align} is known exactly. If $U$ satisfies
\begin{align}U\ket{0^{2n}} &= \ket{\text{Bell}(\mathbb{C}^{2^n})}\\ &= \beta \ket{\text{Bell}(\mathcal{H})} + \sqrt{1-\beta^2}\ket{\phi_\perp} \end{align}
for some $\ket{\phi_\perp} \perp \ket{\text{Bell}(\mathbb{C}^{2^n})}$ then define $U'$ such that:
\begin{align}U'\ket{0^{2n+1}} &= \gamma U\ket{0^{2n}}\ket{0} + \sqrt{1-\gamma^2}\ket{0^{2n}}\ket{1}\\
&= \gamma\beta \ket{\text{Bell}(\mathcal{H})}\ket{0} + \sqrt{1-(\gamma\beta)^2}\ket{\psi_\perp} \end{align}
for some $\ket{\phi_\perp} \perp \ket{\text{Bell}(\mathbb{C}^{2^n})}\ket{0}$ where $\gamma$ is the largest number $\leq 1$ such that \begin{align}\sin( (2k+1) \arcsin(\gamma\beta)) = 1\end{align} has a solution where $k$ is a positive integer. Then, if $\theta = \arcsin(\gamma\beta)$ and $\Pi_\mathcal{H}$ is a projection onto the $\mathcal{H}\otimes \text{span}(\ket{0}\bra{0})$ subspace of $\mathbb{C}^{2n+1}$, then we can define a Grover operator $G$ that exactly prepares $\ket{\text{Bell}(\mathcal{H})}$.
\begin{align}
G = U'(I - &2\ket{0^{2n+1}}\bra{0^{2n+1}})(U')^\dagger(I - 2\Pi_\mathcal{H})\\
G^k \ket{\text{Bell}(\mathbb{C}^{2^n})} &= \sin((2k+1) \theta))\ket{\text{Bell}(\mathcal{H})}\ket{0} \nonumber\\ &+ \cos((2k+1) \theta)\ket{\psi_\perp}\\
&= \ket{\text{Bell}(\mathcal{H})}\ket{0}
\end{align}

Since $2^n < 2D$ we have $\beta \in \Omega(1)$ so $k \in O(1)$, so the circuit complexity is dominated by $U$, which can be constructed using $n$ Hadamard gates and $n$ CNOT gates. Thus the state $I/D$ on a Hilbert space $\mathcal{H}$ encoded in $\mathbb{C}^n$ has an $O(\log(D))$-preparation-unitary. \\

The algorithm for estimating integrals is as follows:

\patbox{
\textbf{Algorithm: Integral of the Density of States}
\begin{enumerate}
    \item Use Theorem~\ref{thm:windowfunc} to construct the polynomial $w(x)$ with $\eta := \frac{\eps}{3}$.
    \item Use Lemma~\ref{lemma:svt} to construct an $\frac{\eps}{3}$-accurate $\frac{1}{2}$-scaled block-encoding of $w(H/\alpha)$. Say that this is an exact $\frac{1}{2}$-scaled block-encoding of $\tilde w(H/\alpha)$.
    \item Use Lemma~\ref{lemma:observ} to produce an $\frac{\eps}{3}$-accurate estimate $\xi$ of $\text{Tr}\left(\frac{I}{D} \cdot \tilde w(H/\alpha) \right)$ with probability at least $(1-\delta)$.
\end{enumerate}
}

By the triangle inequality the total error is at most $\eps$.
The polynomial $w(x)$ has degree:
\begin{align}
d \in  O\left( \frac{\rho_\text{max}}{\eps} \log \frac{\rho_\text{max}}{\eps} \right)
\end{align}
The approximate block-encoding of $\tilde w(H/\alpha)$ has circuit complexity $O(dQ)$ and the preparation unitary for $I/D$ has circuit complexity $\log D$. Combining these with the number of samples required by Lemma~\ref{lemma:observ} gives the overall complexity (\ref{eqn:densintcomplexity}).

The algorithm for Chebyshev Moments is significantly simpler:

\patbox{
\textbf{Algorithm: Chebyshev Moments of Density of States}
\begin{enumerate}
    \item Use Lemma~\ref{lemma:cheby} to construct a block-encoding of $T_n(H/\alpha)$.
    \item Use Lemma~\ref{lemma:observ} to produce an $\eps$-accurate estimate $\zeta$ of $\text{Tr}\left(\frac{I}{D} \cdot T_n(H/\alpha)\right)$ with probability at least $(1-\delta)$.
\end{enumerate}
}

Since the block-encoding and state preparation are exact, the error stems entirely from the estimation procedure in Lemma~\ref{lemma:observ}. The circuit complexity from Lemma~\ref{lemma:cheby} is $O(nQ)$, so the overall complexity (\ref{eqn:denschebycomplexity})  also follows from Lemma~\ref{lemma:observ}.
\end{proof}

Estimation of integrals of $\rho(E)$ benefit from knowledge of an upper bound $\rho_\text{max}$. Indeed even in pathological cases where $H \propto I$ we have $\rho_\text{max} = 1$, so the circuit complexity can never suffer from high densities of state. We argue that in practical situations prior information on $H$ can be used to bound $\rho_\text{max}$, thereby improving the complexity. For example, the DOS of quantum many body systems with local interactions is often close to a Gaussian due to the central limit theorem. In particular, \cite{0406100} discusses the DOS of a nearest-neighbor Hamiltonian acting on a spin chain. From their work on the transverse-field Ising model with $n$ sites we can derive:
$$\rho_\text{max} = \frac{C}{D}\binom{n}{n/2} \approx C\pi\sqrt{\frac{2}{n}}  $$
for some constant $C$ (see the discussion surrounding equation 30 in \cite{0406100}). Here $\rho_\text{max}$ decreases with the number of sites.

Furthermore, exact degeneracy in a Hamiltonian is connected to the Hamiltonian's symmetries \cite{1608.02600}. If there exists a degenerate subspace of dimension $D\rho_\text{max}$ then any unitary transformations on that subspace must preserve the Hamiltonian. Thus, prior knowledge of the symmetries could be used to obtain a bound on $\rho_\text{max}$. However, if only a subset of the symmetries is known then this only leads to a lower bound on the dimension of the largest eigenspace, which is not useful here.

Of course, the efficiency of the algorithm relies on the $1/D$ factor in our definition of $\rho(E)$. If we were interested in the actual number of states within an interval, the circuit complexity would scale with $D$ (for fixed $\eps$). This is to be expected since the number of states in the ground space of a Hamiltonian is \#$\mathsf{P}$-hard to compute exactly and $\mathsf{NP}$-hard to estimate to within relative error \cite{1010.3060}.

Next we consider the local density of states. Say we are working with a Hamiltonian describing a single particle in real space or some space with a notion of locality so that for every position $\vec r$ there is a state $\ket{\psi(\vec r)}$ denoting the state with the particle at $\vec r$. Then local density of states (LDOS) at $\vec r$ is given by \cite{0504627, diventra, 1309.5730}:
\begin{align} \rho_{\vec r}(E) = \sum_i \delta(E_i - E) |\braket{\psi_i | \vec r}|^2  \end{align}

The algorithms for sketching the LDOS are a simple modification of the algorithms for DOS: instead of preparing a maximally mixed state we simply prepare $\ket{\psi(\vec r)}$. Indeed if $\ket{\psi(\vec r)}$ has an $O(R)$-preparation unitary, the new circuit complexities are the same as those in Theorem~\ref{thm:dosalg} but with $\log D$ replaced with $R$.

If $H$ is a lattice Hamiltonian, e.g. a Fermi-Hubbard model, then the states $\ket{\psi(\vec r)}$ are trivial to prepare since the Jordan-Wigner transformation that maps $H$ to qubits preserves locality. For Hamiltonians describing a particle in real-space, the cost of preparing $\ket{\psi(\vec r)}$ depends on the particular choice of basis functions, e.g. Hartree-Fock, used to encode $H$ on the quantum computer.

Similarly to the DOS, estimation of LDOS can benefit from bounds on $\rho_\text{max}$ and it remains true that even for pathological Hamiltonians like $H \propto I$ we have $\rho_\text{max} \leq 1$. However, it no longer makes sense to bound $\rho_\text{max}$ via a central limit theorem since there is only one particle involved. 

\section{Linear Response}

In this section we show how to sketch correlation functions of the form:
\begin{align} A(E - E_0) = \left\langle  B \delta(E - H)C  \right\rangle \label{eqn:lres} \end{align}

 We shift the function by the ground state energy $E_0$ since we consider estimation of the ground state energy out of scope. This work improves on an quantum algorithm by \cite{1804.01505} and is useful to compare to a classical algorithm based on matrix product states \cite{1101.5895} that also uses the kernel polynomial method.

Following a similar argument to (\ref{eqn:windowderiv1}-\ref{eqn:windowderiv5}) and (\ref{eqn:mudef}), we connect the desired quantities to expectations of observables that can be represented by block-encodings:
\begin{align}
    \int_a^b A(E-E_0) dE &\approx \left\langle B w(H/\alpha) C \right\rangle\\
    \mu^A_n &= \left\langle B T_n(H/\alpha) C \right\rangle
\end{align}

This naturally yields quantum algorithms quite similar to those presented in Theorem~\ref{thm:dosalg}, just with some constants changed.

\begin{theorem} Let:
\begin{itemize}
    \item $H$ have an $\alpha$-scaled $Q$-block-encoding,
    \item $\rho$ have an $R$-preparation-unitary,
    \item $B$ have $\beta$-scaled $S_B$-block-encoding and $C$ have $\gamma$-scaled $S_C$-block-encoding.
\end{itemize}
\hspace{1cm}

Then for any $\eps,\delta > 0$:
\begin{enumerate}
 \item For any $a,b$ such that $-\alpha < a < b < \alpha$ there exists a quantum algorithm that produces an estimate $\xi$ of $\int_{a}^b A(E) dE$ with circuit complexity
 \begin{align}O\left( \left( Q d + S_B + S_C + R\right) \cdot \frac{\beta\gamma}{\eps} \log\frac{1}{\delta} \right) \label{eqn:respintcomplexity}\end{align}
 and $O(\text{poly}(d))$ classical pre-processing, where $\rho_\text{max}$ a bound on the dimension of the largest eigenspace and 
 \begin{align}d =O\left( \frac{\rho_\text{max} \beta\gamma}{\eps}\log\frac{\rho_\text{max} \beta\gamma}{\eps}\right).\end{align}
 \item For any $n$ there exists a quantum algorithm that produces an estimate $\zeta$ of $\mu^A_n$ with circuit complexity
 \begin{align}O\left((Qn + S_B + S_C + R)\cdot \frac{\beta\gamma}{\eps}  \right).\label{eqn:respchebycomplexity}\end{align}
\end{enumerate}
The estimates $\xi$ and $\zeta$ have error $\eps$ with probability at least $(1-\delta)$ in their real and imaginary parts.\label{thm:respalg}

\end{theorem}
\begin{proof} The algorithm for computing integrals is as follows:

\patbox{
\textbf{Algorithm: Integrals of Linear Response Functions}
\begin{enumerate}
 \item Use Theorem~\ref{thm:windowfunc} to construct the polynomial $w(x)$ with $\eta := \frac{\eps}{3}$.
 \item Use Lemma~\ref{lemma:svt} to construct an $\frac{\eps}{3}$-accurate $\frac{1}{2}$-scaled block-encoding of $w(H/\alpha)$, and say it is an exact $\frac{1}{2}$-scaled block-encoding of $ \tilde w(H/\alpha)$.
 \item Use Lemma~\ref{lemma:lcu} to construct a $\frac{1}{2}\beta\gamma$-scaled block-encoding of $\Xi := B \tilde w(H/\alpha) C$.
 \item Use Lemma~\ref{lemma:observ} to produce an $\frac{\eps}{3}$-accurate estimates of the real and imaginary parts of $\xi$ with probability at least $(1-\delta)$, corresponding to the Hermitian and anti-Hermitian parts of $\Xi$ as in (\ref{eqn:herm},\ref{eqn:antiherm}).

\end{enumerate}
}

The accuracy and complexity analysis is almost identical to that in Theorem~\ref{thm:dosalg}, except for the fact that since $|B| \leq \beta$ and $|C| \leq \gamma$ we observe that $A(\alpha  x) $ is bounded by $\rho_\text{max}\beta\gamma$ when invoking Theorem~\ref{thm:windowfunc}. The algorithm for Chebyshev moments is as follows:
 
 \patbox{
\textbf{Algorithm: Chebyshev Moments of Linear Response Functions}
\begin{enumerate}
 \item Use Lemma~\ref{lemma:cheby} to construct a block-encoding of $T_n(H/\alpha)$.
 \item Use Lemma~\ref{lemma:lcu} to construct a $\beta\gamma$-scaled block-encoding of $Z := B T_n(H/\alpha)C$.
 \item Use Lemma~\ref{lemma:observ} to produce an $\eps$-accurate estimates of the real and imaginary parts of $\zeta$ with probability at least $(1-\delta)$, corresponding to the Hermitian and anti-Hermitian parts of $Z$ as in (\ref{eqn:herm},\ref{eqn:antiherm}).
\end{enumerate}
}
 
\end{proof}

This technique is significantly more versatile than that of \cite{1804.01505}, which only treats the case when $B = C$ and when $\rho = \ket{\psi_0}\bra{\psi_0}$. Their algorithm runs Hamiltonian simulation under $B$ for a short amount of time to approximately prepare the state $B\ket{\psi_0}$, which is an additional source of error. Furthermore their work also does not capitalize on accuracy improvements from amplitude estimation.

The classical strategy \cite{1101.5895} relies on Matrix Product State (MPS) representations of states $\ket{t_n} = T_n(H/\alpha)C\ket{\psi_0}$. When accurate and efficient MPS representations of $\ket{t_n}$ exist (and $\ket{\psi_0}$ can be efficiently obtained - an assumption we also make), then quantum strategies are not needed. Indeed for many physical systems ground states obey area laws (see e.g. \cite{1905.11337}), which lends MPS strategies their power. Quantum strategies will still be useful for ground states with large amounts of entanglement where efficient classical representations do not exist.

\section{Conclusion}

We have demonstrated that block-encodings provide a powerful framework for the matrix arithmetic on a quantum computer. This modern and versatile toolkit for quantum algorithms encompasses fundamental strategies such as amplitude amplification and estimation, and novel results in active areas like Hamiltonian simulation can be immediately leveraged due to its modularity. Furthermore, once all the necessary tools are assembled, algorithms based on block-encodings are trivial to analyze. We believe that block-encodings are the state-of-the-art technique for estimating physical quantities on a quantum computer. This claim should be further tested by attempting to quantize other numerical strategies in condensed matter physics. 

\section{Acknowledgements}

The author thanks Andras Gilyen, Andrew Potter,  Justin Thaler, Chunhao Wang, Alexander Weisse and Alessandro Roggero for helpful discussions. This work was supported by Scott Aaronson's Vannevar Bush Faculty Fellowship.

This manuscript reflects changes from peer review.

\clearpage
\appendix
\section{Constructing a Polynomial Approximation of the Window Function\label{section:window}}

\vspace{-2mm}
In this section we prove Theorem~\ref{thm:windowfunc} by following a construction in \cite{0902.3757}. We make use of an important theorem in approximation theory:
\vspace{-2mm}

\begin{theorem*} (Jackson's Theorem \cite{rivlin}.) For any continuous function $g(x)$ on the interval $[-1,1]$ there exists a polynomial $J(x)$ of degree at most $n$ so that for all $x \in [-1,1]$:
\vspace{-2mm}
\begin{align} |J(x) - g(x)| \leq 6\omega_g(1/n), \end{align}
\vspace{-2mm}
where $\omega_g(\delta)$ is the modulus of continuity of $g(x)$:
\vspace{-1mm}
\begin{align}\omega_g(\delta)  := \text{sup}\big\{& \text{ } |g(x) - g(y)|\nonumber\\ &\text{ for } x,y \in [-1,1] \text{ with } |x-y| \leq \delta \big\}. \end{align}
\end{theorem*}
\vspace{-3mm}

 Below we prove Theorem~\ref{thm:windowfunc} with $\eta$ rescaled to $\eta f_\text{max}$. If Jackson's theorem were to be used to construct the desired polynomial approximation directly then the degree would scale with $O(\eta^{-2})$. By introducing an amplifying polynomial we improve this to $O( \frac{1}{\eta}\ln\frac{1}{\eta}  )$.

\begin{theorem*} (Theorem~\ref{thm:windowfunc} restated.) For every $\eta > 0$ and any $\bar a, \bar b$ with $-1 < \bar a < \bar b < 1$ there there exists a polynomial $w(x)$ such that for all $f(x)$ bounded by $f_\text{max}$:
\begin{align}\left|\int_{-1}^{1} f(x) w(x) dx - \int_{\bar a}^{\bar b} f(x) dx \right| \leq \eta f_\text{max} \end{align}
The polynomial has degree $d \in O( \frac{1}{\eta}\ln\frac{1}{\eta}  ) $ and $w(x)/2$ satisfies the requirements of Lemma~\ref{lemma:svt}.
\end{theorem*}
\begin{proof} 

Let $\kappa := \eta/4$. We begin by applying Jackson's theorem to a function $g(x)$ sketched in FIG.~1~a). $g(x) = 1$ in the region $[\bar a, \bar b]$ and $g(x) = -1$ outside of $[\bar a - \kappa ,\bar b + \kappa]$ and interpolates linearly between the gaps. We have $\omega_g(\delta) = \delta/\kappa$, so if we choose $n := 24/\kappa$ we obtain:
\begin{align} |J(x) - g(x)| \leq 6\omega_g(1/n) = \frac{6}{\kappa n} = \frac{1}{4} \end{align}
$J(x)$ is sketched in FIG.~1~b), and is guaranteed to stay inside the shaded region. Next we define the amplifying polynomial $A_k(x)$:
\vspace{-2mm}
\begin{align} A_k(x) := \sum_{j \geq k/2} \binom{k}{j} \left(\frac{1+x}{2}\right)^j \left(\frac{1-x}{2}\right)^{k-j} \end{align}
\vspace{-3mm}

Let $X$ be a random variable distributed as the sum of $k$ i.i.d. Bernoulli random variables, each with expectation $\frac{1+x}{2}$, and observe that $A_k(x) = \text{Pr}[X \geq k/2]$. Then it follows from the Chernoff bound that $A_k(x)$ stays inside the shaded region of FIG.~1~c) where $\tau := e^{-k/6}$. Pick $k := \left\lceil 6\ln \frac{4}{\eta} \right\rceil$ so that $\tau \leq \eta/4$.

Finally, we use $A_k(x)$ to amplify the error of $J(x)$.
\begin{align}
w(x) := A_k\left( \frac{4}{5} J(x) \right)
\end{align}
This polynomial $w(x)$ is inside the shaded region of FIG.~1~d) and has degree:
\begin{align}d := n \cdot k \in O\left( \frac{1}{\eta}\ln\frac{1}{\eta}  \right)\end{align}
Now we bound the error, which is intuitive from FIG.~1~d). In the region inside $[\bar a, \bar b]$ and outside $[\bar a - \kappa , \bar b + \kappa]$ we have an error at most $\tau$ and inside the interpolation regions we have error at most $1$. The regions have length $2-2\kappa$ and $2\kappa$ respectively, so:
\begin{align}&\left|\int_{-1}^{1} f(x) w(x) dx - \int_{\bar a}^{\bar b} f(x) dx \right|\cdot \frac{1}{f_\text{max}} \nonumber\\
&\leq \tau(2-2\kappa) + 2\kappa \leq 2\tau + 2\kappa \leq \frac{\eta}{2} + \frac{\eta}{2} = \eta 
\end{align}
    Here we implicitly use (\ref{eq:boundedby}). Note that a more careful choice of the division of error between regions may improve $d$ by a constant factor.
\end{proof}

\onecolumngrid

\begin{figure}[b]
\vspace{-3mm}
\centering
\includegraphics[width=\textwidth]{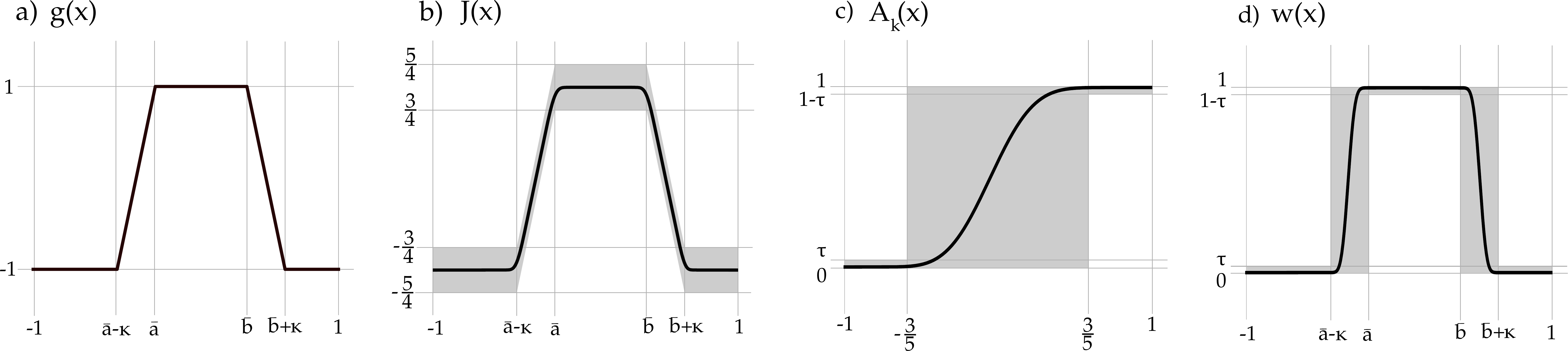}
\vspace{-5mm}
\caption{Functions involved in the proof of Theorem~\ref{thm:windowfunc}. a) The function $g(x)$ which is fed as input to Jackson's theorem. b) The polynomial $J(x)$ returned by Jackson's theorem that is within $1/4$ of $g(x)$. c) The amplifying polynomial $A_k(x)$ is guaranteed to amplify $3/5$ to $\eta := e^{-k/6}$. d) The window function polynomial $w(x)$ and is guaranteed to be inside the shaded region.}
\vspace{5mm}
\end{figure}


\begin{thebibliography}{1}

\bibitem{1806.01838} András Gilyén, Yuan Su, Guang Hao Low, Nathan Wiebe \textbf{Quantum singular value transformation and beyond: exponential improvements for quantum matrix arithmetics}  \href{https://arxiv.org/abs/1806.01838}{arXiv:1806.01838} \href{https://dl.acm.org/doi/10.1145/3313276.3316366}{Proceedings of the 51st Annual ACM SIGACT Symposium on Theory of Computing (STOC 2019) Pages 193–204} (2018)

\bibitem{1707.05391} Guang Hao Low, Isaac L. Chuang \textbf{Hamiltonian Simulation by Uniform Spectral Amplification'} \href{https://arxiv.org/abs/1707.05391}{arXiv:1707.05391} (2017)

\bibitem{2003.02831} Rui Chao, Dawei Ding, Andras Gilyen, Cupjin Huang, Mario Szegedy \textbf{Finding Angles for Quantum Signal Processing with Machine Precision} \href{https://arxiv.org/abs/2003.02831}{arXiv:2003.02831} (2020)


\bibitem{1501.01715}  Dominic W. Berry, Andrew M. Childs, Robin Kothari \textbf{Hamiltonian simulation with nearly optimal dependence on all parameters} \href{https://arxiv.org/abs/1501.01715}{arXiv:1501.01715} \href{https://ieeexplore.ieee.org/document/7354428}{Proceedings of the 56th IEEE Symposium on Foundations of Computer Science (FOCS 2015), pp. 792-809 } (2015)

\bibitem{1906.07115} Dominic W. Berry, Andrew M. Childs, Yuan Su, Xin Wang, Nathan Wiebe \textbf{Time-dependent Hamiltonian simulation with $L^1$-norm scaling} \href{https://arxiv.org/abs/1906.07115}{arXiv:1906.07115} (2019)

\bibitem{1805.00582} Maria Kieferova, Artur Scherer, Dominic Berry \textbf{Simulating the dynamics of time-dependent Hamiltonians with a truncated Dyson series} \href{https://arxiv.org/abs/1805.00582}{arXiv:1805.00582} \href{https://journals.aps.org/pra/abstract/10.1103/PhysRevA.99.042314}{Phys. Rev. A 99, 042314} (2018)

\bibitem{1805.00675} Guang Hao Low, Nathan Wiebe \textbf{Hamiltonian Simulation in the Interaction Picture} \href{https://arxiv.org/abs/1805.00675}{arXiv:1805.00675} (2018)


\bibitem{1511.02306} Andrew M. Childs, Robin Kothari, Rolando D. Somma \textbf{Quantum algorithm for systems of linear equations with exponentially improved dependence on precision} \href{https://arxiv.org/abs/1511.02306}{arXiv:1511.02306} \href{https://epubs.siam.org/doi/10.1137/16M1087072}{SIAM Journal on Computing 46, 1920-1950} (2015)

\bibitem{1603.02940} Anirban Narayan Chowdhury, Rolando D. Somma \textbf{Quantum algorithms for Gibbs sampling and hitting-time estimation} \href{https://arxiv.org/abs/1603.02940}{arXiv:1603.02940} \href{https://dl.acm.org/doi/10.5555/3179483.3179486}{Quant. Inf. Comp. Vol. 17, No. 1/2, pp. 0041-0064} (2016)

\bibitem{1610.06546} Guang Hao Low, Isaac L. Chuang \textbf{Hamiltonian Simulation by Qubitization} \href{https://arxiv.org/abs/1610.06546}{arXiv:1610.06546} \href{https://quantum-journal.org/papers/q-2019-07-12-163/}{Quantum 3, 163} (2016)

\bibitem{1606.02685} Guang Hao Low, Isaac L. Chuang \textbf{Optimal Hamiltonian Simulation by Quantum Signal Processing} \href{https://arxiv.org/abs/1606.02685}{arXiv:1606.02685} \href{https://journals.aps.org/prl/abstract/10.1103/PhysRevLett.118.010501}{Phys. Rev. Lett. 118, 010501} (2016)

\bibitem{1409.3305} Theodore J. Yoder, Guang Hao Low, Isaac L. Chuang \textbf{Fixed-point quantum search with an optimal number of queries} \href{https://arxiv.org/abs/1409.3305}{arXiv:1409.3305} \href{https://journals.aps.org/prl/abstract/10.1103/PhysRevLett.113.210501}{Phys. Rev. Lett. 113, 210501} (2014)

\bibitem{0902.3757} Ilias Diakonikolas, Parikshit Gopalan, Ragesh Jaiswal, Rocco Servedio, Emanuele Viola \textbf{Bounded Independence Fools Halfspaces} \href{https://arxiv.org/abs/0902.3757}{arXiv:0902.3757} \href{https://dl.acm.org/doi/10.5555/1747597.1748034}{Proceedings of the 2009 50th Annual IEEE Symposium on Foundations of Computer Science (FOCS 2015) Pages 171–180}  (2009)

\bibitem{rivlin} Theodore J. Rivlin \textbf{An Introduction to the Approximation of Functions} Dover Publications, Inc. New York. (1969)

\bibitem{0604324} Alexandre Eremenko, Peter Yuditskii \textbf{Uniform approximation of $\text{sgn}(x)$ by polynomials and entire functions} \href{https://arxiv.org/abs/math/0604324}{arXiv:0604324} \href{https://link.springer.com/article/10.1007/s11854-007-0011-3}{Journal d'Analyse Math\'ematique, 101 313-324} (2006)

\bibitem{dolph} C. L. Dolph. \textbf{A Current Distribution for Broadside Arrays Which Optimizes the Relationship Between Beam Width and Side-Lobe Level} \href{https://ieeexplore.ieee.org/document/1697083}{Proceedings of the IRE, 34(6):335–348}, 1946.

\bibitem{0406180} Julia Kempe, Alexei Kitaev, Oded Regev \textbf{The Complexity of the Local Hamiltonian Problem}
\href{https://arxiv.org/abs/quant-ph/0406180}{arXiv:0406180} \href{https://epubs.siam.org/doi/abs/10.1137/S0097539704445226?mobileUi=0}{SIAM Journal of Computing, Vol. 35(5), p. 1070-1097} (2004)

\bibitem{1212.6312} Adam D. Bookatz. \textbf{QMA-complete problems} \href{https://arxiv.org/abs/1212.6312}{arXiv:1212.6312} \href{https://dl.acm.org/doi/abs/10.5555/2638661.2638662}{Quantum Information and Computation, Vol.14 No.5-6} (2012)

\bibitem{1010.3060} Brielin Brown, Steven T. Flammia, Norbert Schuch \textbf{Computational Difficulty of Computing the Density of States} \href{https://arxiv.org/abs/1010.3060}{arXiv:1010.3060} \href{https://journals.aps.org/prl/abstract/10.1103/PhysRevLett.107.040501}{Phys. Rev. Lett. 107, 040501} (2010)

\bibitem{1608.02600} Christopher T. Chubb, Steven T. Flammia \textbf{Approximate symmetries of Hamiltonians} \href{https://arxiv.org/abs/1608.02600}{arXiv:1608.02600} \href{https://aip.scitation.org/doi/10.1063/1.4998921}{Journal of Mathematical Physics 58, 082202} (2016)


\bibitem{1905.11337} Anurag Anshu, Itai Arad, David Gosset \textbf{Entanglement subvolume law for 2D frustration-free spin systems} \href{https://arxiv.org/abs/1905.11337}{arXiv:1905.11337} (2019)

\bibitem{1609.07877} Fernando G.S.L. Brandao, Michael J. Kastoryano \textbf{Finite correlation length implies efficient preparation of quantum thermal states} \href{https://arxiv.org/abs/1609.07877}{arXiv:1609.07877} \href{https://link.springer.com/article/10.1007\%2Fs00220-018-3150-8}{Communications in Mathematical Physics 365.1: 1-16} (2016)

\bibitem{1101.5895} Andreas Holzner, Andreas Weichselbaum, Ian P. McCulloch, Ulrich Schollwöck, Jan von Delft \textbf{Chebyshev matrix product state approach for spectral functions} \href{https://arxiv.org/abs/1101.5895}{arXiv:1101.5895} \href{https://journals.aps.org/prb/abstract/10.1103/PhysRevB.83.195115}{Phys. Rev. B 83, 195115} (2011)

\bibitem{0504627} Alexander Weisse, Gerhard Wellein, Andreas Alvermann, Holger Fehske \textbf{The Kernel Polynomial Method} \href{https://arxiv.org/abs/cond-mat/0504627}{arXiv:0504627} \href{https://journals.aps.org/rmp/abstract/10.1103/RevModPhys.78.275}{Rev. Mod. Phys. 78, 275} (2005)

\bibitem{1811.07387} Zheyong Fan, Jose Hugo Garcia, Aron W. Cummings, Jose Eduardo Barrios-Vargas, Michel Panhans, Ari Harju, Frank Ortmann, Stephan Roche \textbf{Linear Scaling Quantum Transport Methodologies} \href{https://arxiv.org/abs/1811.07387}{arXiv:1811.07387} \href{https://journals.aps.org/prb/abstract/10.1103/PhysRevB.89.045123}{Phys. Rev. B 89, 045123} (2018)

\bibitem{1309.5730}  Elahe Yeganegi, Ad Lagendijk, Allard P. Mosk, Willem L. Vos \textbf{Local density of optical states in the band gap of a finite photonic crystal} \href{https://arxiv.org/abs/1309.5730}{arXiv:1309.5730} \href{https://journals.aps.org/prb/abstract/10.1103/PhysRevB.89.045123}{Phys. Rev. B 89, 045123} (2013)

\bibitem{0406100} Michael Hartmann, Guenter Mahler, Ortwin Hess \textbf{Spectral densities and partition functions of modular quantum systems as derived from a central limit theorem} \href{https://arxiv.org/abs/cond-mat/0406100}{arXiv:0406100} \href{https://link.springer.com/article/10.1007\%2Fs10955-004-4298-5}{Journal of Statistical Physics volume 119, pages1139–1151} (2004)

\bibitem{diventra} Massimilano Di Ventra \textbf{Electrical Transport in Nanoscale Systems} Cambridge University Press, New York. (2008)

\bibitem{1908.10846} Scott Aaronson, Patrick Rall \textbf{Quantum Approximate Counting, Simplified} \href{https://arxiv.org/abs/1908.10846}{arXiv:1908.10846} \href{https://epubs.siam.org/doi/abs/10.1137/1.9781611976014.5}{Symposium on Simplicity in Algorithms. 2020, 24-32} (2019)

\bibitem{1912.05559} Dmitry Grinko, Julien Gacon, Christa Zoufal, Stefan Woerner \textbf{Iterative Quantum Amplitude Estimation} \href{https://arxiv.org/abs/1912.05559}{arXiv:1912.05559} (2019)

\bibitem{0005055} Gilles Brassard, Peter H{\o}yer, Michele Mosca, Alain Tapp \textbf{Quantum Amplitude Amplification and Estimation} \href{https://arxiv.org/abs/quant-ph/0005055}{arXiv:0005055} Quantum Computation and Information, Contemporary Mathematics Series. AMS. (2000)

\bibitem{1401.2430} J. S. Pedernales, R. Di Candia, I. L. Egusquiza, J. Casanova, E. Solano \textbf{Efficient Quantum Algorithm for Computing n-time Correlation Functions} \href{https://arxiv.org/abs/1401.2430}{arXiv:1401.2430} \href{https://journals.aps.org/prl/abstract/10.1103/PhysRevLett.113.020505}{Phys. Rev. Lett. 113, 020505} (2014)
\bibitem{1804.01505} Alessandro Roggero, Joseph Carlson \textbf{Linear Response on a Quantum Computer} \href{https://arxiv.org/abs/1804.01505}{arXiv:1804.01505} \href{https://journals.aps.org/prc/abstract/10.1103/PhysRevC.100.034610}{Phys. Rev. C 100, 034610} (2018)

\bibitem{1907.11748} Rolando D. Somma \textbf{Quantum eigenvalue estimation via time series analysis} \href{https://arxiv.org/abs/1907.11748}{arXiv:1907.11748} \href{https://iopscience.iop.org/article/10.1088/1367-2630/ab5c60}{New J. Phys. 21 123025} (2019)

\bibitem{2004.04889} Alessandro Roggero \textbf{Spectral density estimation with the Gaussian Integral Transform} \href{https://arxiv.org/abs/2004.04889}{arXiv:2004.04889} (2020)

\bibitem{west} G. B. West \textbf{Electron scattering from atoms, nuclei and nucleons}  \href{https://www.sciencedirect.com/science/article/abs/pii/0370157375900356}{Physics Reports \textbf{18}, 263} (1975)

\bibitem{sears} V. F. Sears. \textbf{Scaling and final-state interactions in deep-inelastic neutron scattering} \href{https://journals.aps.org/prb/abstract/10.1103/PhysRevB.30.44}{Phys. Rev. B \textbf{30}, 44} (1984)

\end{thebibliography}
\end{document}